\documentclass[journal,onecolumn]{IEEEtran}

\usepackage{graphicx} % Required for the inclusion of images
\usepackage{multirow}
\usepackage{amsmath}
\usepackage{mathrsfs}
\usepackage[utf8]{inputenc}
\usepackage[english]{babel}
\usepackage{algorithm}
\usepackage{algpseudocode}
\usepackage{amsthm}
\usepackage{dsfont}
\usepackage{amssymb}
\usepackage[usenames,dvipsnames]{color}
\usepackage{mathtools}

\newtheorem{theorem}{Theorem}

\newtheorem{lemma}{Lemma}
\newtheorem{definition}{Definition}
\newtheorem{remark}{Remark}

\newtheorem{example}{Example}

\newtheorem*{example*}{Example 1 revisited}

\newtheorem{proportion}{Proportion}

\begin{document}

\title{Single-Server Multi-Message Private Information Retrieval with Side Information}
\author{Su Li and Michael Gastpar\\EPFL}

\maketitle

\begin{abstract}
We study the problem of single-server multi-message private information retrieval with side information. 
One user wants to recover $N$ out of $K$ independent messages which are stored at a single server. The user initially possesses a subset of $M$ messages as side information.
The goal of the user is to download the $N$ demand messages while not leaking any information about the indices of these messages to the server.
In this paper, we characterize the minimum number of required transmissions. 
We also present the optimal linear coding scheme which enables the user to download the demand messages and preserves the privacy of their indices.
Moreover, we show that the trivial MDS coding scheme with $K-M$ transmissions is optimal if $N>M$ or $N^2+N \ge K-M$. 
This means if one wishes to privately download more than the square-root of the number of files in the database, then one must effectively download the full database (minus the side information), irrespective of the amount of side information one has available.

\end{abstract}

%% Note that keywords are not normally used for peerreview papers.
%\begin{IEEEkeywords}
%IEEE, IEEEtran, journal, \LaTeX, paper, template.
%\end{IEEEkeywords}

\IEEEpeerreviewmaketitle

\section{Introduction}
Consider $K$ independent messages stored at a single server. 
One user wants to download $N$ messages from the sever while it already has $M$ messages as side information.
The user sends queries to the server and the server replies with (coded) messages according to the user's requests.
Private information retrieval requires that the server should not be able to infer any information about the indices of the messages that the user wants to download.
We refer to this problem as Single-server Multi-message Private Information Retrieval with Side Information (SMPIRSI).
To solve the SMPIRSI problem, we need to find the minimum required number of transmissions that the user should request from the server and the optimal linear coding scheme which enables the user to decode the demand messages while protecting the indices of the demand messages from the server.

\subsection{Related Work}
The Private Information Retrieval (PIR) problem was first studied in~\cite{492461} from a computational complexity perspective. 
Recently, the PIR problem attracted considerable attention in the information theory society and many works study this problem from an information-theoretic point of view~\cite{1181949,Yekhanin:2010:PIR:1721654.1721674,7889028,8119895}.
A single user wants to privately download one message from a database. 
To achieve perfect privacy in the information-theoretic sense, if the database is only stored at one server, the user has to download all messages.
The problem becomes more interesting if one supposes that the database is stored in multiple servers and there is no collusion between these servers.
By exploiting the advantages of replications of the database in non-colluding servers, private information retrieval can be achieved without downloading all messages and the capacity of this problem is characterized in~\cite{7889028}.
Ensuing work has studied many variations of this theme, including databases coded by erasure codes~\cite{6874954,7282975,8006763,7997029,7997393,7541531}, partial colluding servers~\cite{8119895,8006861,2017arXiv170601442B}, side information messages available at users~\cite{2017arXiv170900112K,2017arXiv170607035T,2017arXiv170901056W,2017arXiv170903022C,Li1806:Single} and  multiple messages~\cite{8006859,DBLP:journals/corr/abs-1805-11892},

In~\cite{8006859}, Banawan and Ulukus consider the problem that the user wants to download multiple messages from multiple servers, but there is no side information at the user.
In~\cite{DBLP:journals/corr/abs-1805-11892}, Shariatpanahi {\it et al.}  study the multi-message PIR problem with side information and the user wants to protect both the privacy of the indices of demand messages and of the side information messages.
In our problem, the user is only interested in protecting the privacy of the indices of the demand messages, which is a more challenging problem than protecting both the indices of the demand and side information messages.
The single-server multi-message PIR with side information problem is studied concurrently in~\cite{2018arXiv180709908H}, which has the same results as us when $N>M$ and presents achievability results when $N\le M$\footnote{Our work and~\cite{2018arXiv180709908H} study the same problem independently. We submitted our work to the 56th Allerton conference on 9th July 2018.}.

\subsection{Contributions}

\begin{itemize}
	\item[(1)] We present a closed-form expression for the minimum number of required transmissions for SMPIRSI problem.
		
	\item[(2)] We propose a novel method, Partition-and-MDS-Coding, to generate optimal linear coding schemes with satisfy the requirements of SMPIRSI and use the minimum number of transmissions.
	
	\item[(3)] We show that the trivial MDS coding scheme with $K-M$ transmissions is optimal when the number of demand messages satisfies either $N> M$ or $N^2+N \ge K-M$.
\end{itemize}

\section{System Model and Definition}

Consider a server which stores $K$ independent messages, denoted by $\mathbf{X} = \{X_1,\dots,X_K\}$. 
Each message $X_i \in \mathbb{F}$, where $\mathbb{F}$ is some finite field.
One user initially has $M$ side information messages and wants to download $N$ messages from the server, while the user does not want to reveal any information about the indices of the demand messages to the server.
We assume that the server only knows the number $M$ of side information messages of the user but has no idea about which messages the user has.

Let $\mathbf{W} = \{W_1,\dots,W_N\} \subseteq [K]$ denote the set of indices of the demand messages and $\mathbf{S} = \{S_1,\dots,S_M\} \subset [K] \setminus \mathbf{W}$ denote the set of side information messages. 
Let $\mathcal{S}$ and $\mathcal{W}$ denote the random variables corresponding to the indices of side information messages and demand messages.
We assume that $\mathcal{W}$ is uniformly distributed over all subsets of $[K]$ with size $N$, i.e., 
\begin{align}
	\Pr(\mathcal{W} = \mathbf{W}) = \frac{1}{\binom{K}{N}}&& \forall \mathbf{W} \subseteq [K], |\mathbf{W} |=N.
\end{align}
Moreover, $\mathcal{S}$ is also uniformly distributed over all subsets of $[K]\setminus \mathbf{W}$ with size $M$, i.e.,
\begin{align}
	\Pr(\mathcal{S} = \mathbf{S}| \mathbf{W}) = \frac{1}{\binom{K-N}{M}} && \forall \mathbf{S} \subseteq [K]\setminus \mathbf{W}, |\mathbf{S}| = M.
\end{align}

To retrieve the demand messages, the user sends a query $Q(\mathbf{W},\mathbf{S})$, which is determined by the indices of the demand messages and side information messages, to the server and the server replies coded messages according to the query.
In this paper, we only consider linear coding schemes.
Let $\mathbf{T} = \{T_1,\dots,T_R\}$ denote the linear coding scheme with $R$ transmissions to be sent to the user by the server.
Private information retrieval requires $\mathbf{T}$ to satisfy two conditions:
\begin{enumerate}
	\item \textbf{Retrieval Condition (Correctness)}: The user should be able to decode all demand messages from $\mathbf{T}$ by using its locally available side information messages, that is,
	\begin{align}
		H(X_{\mathbf{W}} | \mathbf{T}, X_{\mathbf{S}}) = 0\label{eq:retcon}.
	\end{align}
	\item \textbf{Privacy Condition}: The server should not be able to infer any information about the indices of the demand messages from the query, that is,
	\begin{align}
		I(\mathcal{W};Q(\mathbf{W},\mathbf{S})) = 0\label{eq:pricon}.
	\end{align}
\end{enumerate}
To satisfy the \textbf{Privacy Condition}, it is equivalent to have 
\begin{align}
	H(\mathcal{W}|Q(\mathbf{W},\mathbf{S})) = H(\mathcal{W})
\end{align}

\begin{definition}[\textbf{Coding subspace}]\label{def:codingsubspace}
	For any linear coding scheme $\mathbf{T} = \{T_1,\dots,T_R\}$, let $\text{supp}(T_i)$ denote the messages which are used to generate $T_i$.
	Define a partition of messages $\mathcal{P}(\mathbf{T}) = \{\wp_1,\dots\}$ such that for each $T_i$, there exists a unique $\wp_j \in \mathcal{P}$ such that $\text{supp}(T_i) \subseteq \wp_j$.
	We call the subspace spanned by each $\wp_j$ a coding subspace.
\end{definition}

For any linear coding scheme, it is possible to find its coding subspace(s).
The coding subspaces should jointly contain all messages, otherwise the non-included message cannot be the demand message, which violates the Privacy Condition.
The minimum number of required transmissions for single-server multi-message private information retrieval with side information can be computed by
\begin{align}
R^* = \min_{\mathcal{L} \in \Pi(K)} \sum_{i =1}^{|\mathcal{L}|} R(\mathcal{L}_i)\label{eq:r*}
\end{align}
where $\Pi(K)$ is the set of all partitions of $K$, $\mathcal{L} = \{\mathcal{L}_1,\dots,\mathcal{L}_{|\mathcal{L}|}\}$ and $R(\mathcal{L}_i)$ is the minimum number of required transmissions for coding subspace with size $\mathcal{L}_i$.

\begin{definition}[\textbf{MDS-Condition}]
	A linear coding scheme $\mathbf{T}$ satisfies the MDS-Condition in the coding subspace spanned by the messages in $\wp_i\in \mathcal{P}(\mathbf{T})$ if there exists a non-negative integer $M_i$ such that 
	given any $M_i$ messages in $\wp_i$, all other messages can be decoded and none of the messages can be decoded given any $M_i-1$ messages.
\end{definition}

\begin{example}\label{eg:1}
	Consider the SMPIRSI problem with the following setup: 
	$\mathbf{X} = \{X_1,X_2,\dots,X_{13}\}$, $\mathbf{W} = \{2,5\}$ and $\mathbf{S} = \{1,4,6,7,9\}$.
	We give the following linear coding scheme $\mathbf{T} = \{T_1,\dots,T_6\}$:
	\begin{align}
		T_1 &= X_1 + X_2+X_4+X_6+X_8\\
		T_2 &= X_1 + 2X_2+3X_4+4X_6+5X_8\\
		T_3 &= X_3+X_{10}+X_{11}+X_{13}\\
		T_4 &= X_3+2X_{10}+3X_{11}+4X_{13}\\
		T_5 &= X_5+X_{7}+X_{9}+X_{12}\\
		T_6 &= X_5+2X_{7}+3X_{9}+4X_{12}
	\end{align}
	
	From the user's perspective: $X_2$ can be decoded from $T_1$ and $T_2$ given that $X_1$, $X_4$ and $X_6$ are side information. $X_5$ can be decoded from $T_5$ and $T_6$ given that $X_7$ and $X_9$ are side information. Hence, the Retrieval Condition is satisfied.
	
	From the server's perspective: The linear coding scheme has $3$ coding subspaces: $\wp_1 = \{X_1,X_2,X_4,X_6,X_8\}$, $\wp_2 = \{X_3,X_{10},X_{11},X_{13}\}$ and $\wp_3 = \{X_5,X_7,X_9,X_{12}\}$.
	They satisfy the MDS-Condition with $m_1= 3$, $m_2=2$ and $m_3= 2$, respectively.
	Since the server only knows that the user has $5$ side information messages and wants to download $2$ messages,
	it can only infer that the demand messages are in either the same coding subspace or in $2$ different coding subspaces. 
	By using the randomized construction process shown in Section~\ref{sec:CodingScheme}, we can show that the probability for any two messages to be the demand messages is the same. 
	Hence, this coding scheme also satisfies the Privacy Condition.

	We will show that for this problem, we need at least $6$ transmissions and present the proof for the reason why it is optimal to partition messages into these three coding subspaces and why it is sufficient to have two transmissions in each coding subspace in Section~\ref{sec:main}.
\end{example}

\section{Main Result}\label{sec:main}
The main result of this paper is presented by the following theorem. We give a closed-form expression for the minimum number of required transmissions for SMPIRSI problem.

\begin{theorem}~\label{Thm:main}
	For the single-server multi-message private information retrieval with side information problem, the minimum number of required transmissions satisfies
	\begin{align}
		R^*(K,M,N)= K-M - (L^*-1-N)^{+}\bar{M}- ((L^*-N)V)^+\label{eq:main}
	\end{align}
	where $\bar{M} = \lfloor\frac{M}{N}\rfloor$, $t = M- N\bar{M}$, $L^* = \left\lceil\frac{K-t}{ \bar{M}+N}\right\rceil$ and $V = (K- (L^*-1)(\bar{M} +N)-t -N)/(L^*-N)$.
\end{theorem}
\begin{remark}
	According to Theorem~\ref{Thm:main}, the minimum number of required transmissions is always upper bounded by $K-M$, which is consistent with the fact that there always exists a linear PIR coding scheme which is an MDS code with $K-M$ transmissions. This code has all messages in a single coding subspace. 
\end{remark}

In order to prove Theorem~\ref{Thm:main}, we first prove some useful lemmas. We also provide an alternative proof in Appendix.

\begin{lemma}\cite{Li1806:Single}\label{lm:mds}
	For any linear coding scheme that satisfies the Privacy Condition, without loss of optimality, the MDS-Condition should be satisfied in every coding subspace. 
\end{lemma}
\begin{proof}
	For a linear coding scheme $\mathbf{T} = \{T_1,\dots,T_R\}$, suppose there exists a coding subspace $\wp \in \mathcal{P}(\mathbf{T})$ such that the MDS-Condition is not satisfied in $\wp$.
	Let $D(\mathbf{T},X)$ denote the minimum number of side information messages which are required to decode message $X$ from linear coding scheme $\mathbf{T}$.
	Then, there must exists two messages $X_i$ and $X_j$ in the same coding subspace $\wp$ such that 
	\begin{itemize}
		\item[(i)] $D(X_i) > D(X_j)$.
		\item[(ii)] $D(X_i) = D(X_j)$, but there are more choices of side information messages for $X_j$ than $X_i$.
	\end{itemize}
	 For the first case, suppose $X_j$ can be decoded from $T_k$ given $D(\mathbf{T},X_i)$ as side information. If $T_k$ is not the transmission that will be used by the user to decode the demand messages, then it is possible to remove messages in $\text{supp}(T_k)$ which are also in $\text{supp}(T_l)$ $\forall l \ne k$.
	 As a result, we either $X_j$ cannot be decoded from $T_k$ or $X_j$ is in another coding subspace.
	 For the second case, we can do similar operation on the transmissions which can be used to decode $X_j$ given $D(\mathbf{T},X_j)$ as side information.
	 In both cases, If the original coding scheme satisfies Privacy condition, the modified coding scheme also satisfies Privacy condition and uses the same number of transmissions.
	 Hence, it is optimal to only consider the linear coding schemes which satisfy MDS-Condition in every coding subspace as the candidate PIR coding scheme.
\end{proof}

%\begin{example}
%	Suppose the user generates a coding scheme which has the following two transmissions in coding subspace $\wp=\{X_1,X_2,X_3,X_4\}$:
%	\begin{align}
%		T_1 = X_1+X_2+X_3 && T_2 = X_1+X_4
%	\end{align}
%	It is easy to verify that MDS-Condition is not satisfied in $\wp$. And $D(X_1) =D(X_4)= 1$ and $D(X_2)=D(X_3) =2$.
%	If $T_2$ is not the transmissions that the user will used to decode its demand messages, then the modified coding scheme can be
%	\begin{align}
%		T_1 = X_1+X_2+X_3 && T_2 = X_4
%	\end{align}
%	the original coding subspace $\wp$ is further decomposed into two smaller coding subspaces where MDS-Condition is satisfied in both coding subspaces.
%\end{example}

According to Lemma~\ref{lm:mds}, it is without loss of optimality to restrict attention to linear coding schemes which satisfy the MDS-Condition in every coding subspace. 
For such coding schemes, the minimum number of required transmissions in each coding subspace should satisfy the condition stated in the following Lemma.

\begin{lemma}\label{lm:tscs}
	For coding subspace $\wp$, the number of transmissions in this coding subspace, $R(\wp)$, satisfies
	\begin{align}
	R(\wp,m(\wp)) \left\{ \begin{aligned}
	&=|\wp|,  & |\wp| \le N\\
	&\ge |\wp|-m(\wp), & N+1\le |\wp|\le K
	\end{aligned}
	\right.		
	\end{align} 
	where $m(\wp) \in \{0,1,\dots, |\wp|-N \}$ is the number of side information messages that are used in such coding subspace.
\end{lemma}
\begin{proof}
	If the coding subspace $\wp$ includes no more than $N$ messages, since all the messages in $\wp$ can be the demand message, the transmissions in such coding subspace should be equal to the dimension of the coding subspace, which is $|\wp|$.
	If the coding subspace $\wp$ includes more than $N$ messages, then it must satisfy the MDS-Condition, according to Lemma~\ref{lm:mds}.
	Hence, by using $m(\wp)$ side information messages, the number of transmissions in such coding subspace should be $|\wp|-m(\wp)$.
	Additionally, since it is possible that all $N$ demand messages are in $\wp$, the number of transmissions in $\wp$ is at least $N$.
	%And we know that the total number of transmissions is upper bounded by $K-M$, which is achieved by MDS coding scheme for all messages.
	Thus, the maximum number of side information messages that can be used in this coding subspace is $|\wp|- N$.
\end{proof}

\begin{example*}
	It can be verified that the coding scheme proposed in Example~\ref{eg:1}, $\mathbf{T}=\{T_1,\dots,T_6\}$, satisfies MDS-Condition in all three coding subspaces, with $m_1 = 3$, $m_2= 2$ and $m_3= 2,$ respectively.
\end{example*}

Let $\mathcal{P} = \{\wp_1,\dots,\wp_L\}$ denote the coding subspaces. Without loss of generality, we assume that $|\wp_1|\ge\dots,|\wp_L|$.
The minimum number of required transmissions of any linear PIR scheme based on such coding subspaces can be computed as
\begin{align}
	R(\mathcal{P}) = \min_{\mathbf{m}}\sum_{i = 1}^{L} R(\wp_i,m_i) 
\end{align}
where $\mathbf{m} = \{m_1,\dots,m_L\}$ is the vector of the number of side information messages used in each coding subspace.
It is easy to see that a coding subspace of larger size should have no fewer side information messages, i.e., $m_1\ge \dots\ge m_L$.
Since the total number of side information messages is $M$, the feasible side information vector should satisfy
\begin{align}
	\sum_{i =1}^{\min\{L,N\}} m_i = M\label{eq:totalsideinformation}
\end{align}
The reason why the summation is taken only from $1$ to $\min\{L,N\}$ is that when $L>N$, the number of coding subspaces which contain demand messages is at most $N$.
If the first $m_i$'s sums up to $M$, every subset of $m_i$'s with size $N$ has sum no larger than $M$.
Hence, Eqn.~\eqref{eq:totalsideinformation} guarantees that the total number of side information messages used by any $N$ coding subspaces is not larger than $M$. 
Additionally, if the size of one coding subspace is equal to or less than $N$, the number of side information messages that can be used in such coding subspace can only be zero.

The optimization problem~\eqref{eq:r*} for the minimum number of required transmissions can be expressed as follows by optimizing over side information vectors:
\begin{align}
	R^* &= \min_{\mathcal{L} \in \Pi(K)} \sum_{i =1}^{|\mathcal{L}|} R(\mathcal{L}_i) \\
	&= \min_{\mathcal{L} \in \Pi(K)} \min_{\mathbf{m} \in \mathbf{P}(M)} \sum_{i =1}^{|\mathcal{L}|} R(\mathcal{L}_i,m_i)\\
	&= \min_{\mathcal{L} \in \Pi(K)} \min_{\mathbf{m} \in \mathbf{P}(M)} \sum_{i =1}^{|\mathcal{L}|} \mathcal{L}_i-m_i\\
	& = K-\max_{\mathcal{L} \in \Pi(K)} \max_{\mathbf{m} \in \mathbf{P}(M)} \sum_{i =1}^{|\mathcal{L}|}m_i
\end{align} 
where $\mathcal{L}=\{\mathcal{L}\_1,\dots,\mathcal{L}_{|\mathcal{L}|}\}$ is a partition of the integer $K$ which satisfies $\forall i \in [|\mathcal{L}|]: \mathcal{L}_i>N$ and $\mathbf{m}=\{m_1,\dots,m_{|\mathcal{L}|}\}$ is the vector of the number of side information messages used in each coding subspace which satisfies Eqn.~\eqref{eq:totalsideinformation} and $m_i \le (\mathcal{L}_i-N)^+$.
Since $\mathcal{L}$ is a partition of $K$, it is always true that the summation $\sum_{i =1}^{|\mathcal{L}|} L_i = K$.
Therefore, the optimization problem becomes find the optimal partition ($\mathcal{L}$) and optimal side information vector $\mathbf{m}$ such that $\sum_{i =1}^{|\mathcal{L}|}m_i$ is maximized.

%\begin{corollary}
%	It is sufficient to only consider the partitions with size between $1$ and $\lfloor \frac{K}{N+1} \rfloor$.
%\end{corollary}
%\begin{proof}
%	According to Lemma~\ref{lm:minisizecodingsubspace},  there exists an optimal partition such that every coding subspace has more than $N$ messages.
%	Hence, the number of coding subspaces is at most $\lfloor \frac{K}{N+1} \rfloor$.
%\end{proof}

\begin{lemma}\label{lm:sizelessthanN}
	For any linear PIR coding scheme with fewer than $N$ coding subspaces, the minimum number of required transmissions is always $K-M$.
\end{lemma}
\begin{proof}
	For any partition $\mathcal{L} = \{\mathcal{L}_1,\dots,\mathcal{L}_L \}$, if $L \le N$, it is possible that every coding subspace contains at least one demand messages.
	Hence, the number of side information messages used in all coding subspace must sum up to $M$.
	%According to Lemma~\ref{lm:minisizecodingsubspace}, we know that every coding subspace can be treated has more than $N$ messages.
	Thus, we have
	\begin{align}
	R^*(\mathcal{L}) = \sum_{i =1}^{L} R(\mathcal{L}_i,m_i) = \sum_{i=1}^{L} (\mathcal{L}_i - m_i) = \sum_{i=1}^{L} \mathcal{L}_i - \sum_{i=1}^{L} m_i = K- M
	\end{align}
\end{proof}

\begin{lemma}\label{lm:sizelargerthanN}
	For any linear PIR coding scheme based on a partition $\mathcal{L} = \{\mathcal{L}_1,\dots,\mathcal{L}_L \}$ with $L >N$ and $\mathcal{L}_1\ge \dots\ge \mathcal{L}_L\ge N$, the minimum number of transmissions satisfies
	\begin{align}
	R(\mathcal{L})^* &= \min_{\mathbf{m}}\sum_{i=1}^{L} R(\mathcal{L}_i) =K-\max_{\mathbf{m}} \sum_{i=1}^{L}m_i = K-M -\max_{\mathbf{m}}\sum_{i=N+1}^{L}m_i
	%&= \sum_{i=1}^{N}R(\mathcal{L}_i) +\sum_{i=N+1}^{L}R(\mathcal{L}_i)  = \sum_{i=1}^{L} \mathcal{L}_i - \sum_{i=1}^{N}m_i - \sum_{i=N+1}^{L}m_i = K-M-\sum_{i=N+1}^{L}m_i
	\end{align}
	where $\mathbf{m} = \{m_1,\dots,m_L\}$ is the vector of the number of side information messages in each coding subspace.
\end{lemma}

Hence, when the number of coding subspaces is larger than $N$,
the number of side information messages that can be used in the coding subspace corresponding to $\mathcal{L}_{N+1},\dots,\mathcal{L}_L$ should be maximized.
Recall that we assume $m_1\ge m_2\ge \dots\ge m_L$, hence, we actually need to maximize $m_N$. 
Since vector $\mathbf{m}=\{m_1,\dots,m_L\}$ satisfies Eqn.~\eqref{eq:totalsideinformation}, we have $m_N \le \lfloor\frac{M}{N}\rfloor$.
Moreover, for coding subspaces with size $\mathcal{L}_i$ such that $N<\mathcal{L}_i<N+\lfloor \frac{M}{N}\rfloor$, the number of side information messages used in such a coding subspace satisfies $m_i \le \mathcal{L}_i-N$, since the number of transmissions used in such a coding subspace is at least $N$.

\begin{lemma}\label{lm:vector4sideinformation}
	Let $\mathbf{m}=\{m_1,\dots,m_L\}$ denote the number of side information messages used in each coding subspace, where $L > N$ and $m_1\ge \dots\ge m_L$.
	Let $\mathcal{L}_i$ denote the size of the $i$-th coding subspace and $t = M- N\lfloor \frac{M}{N} \rfloor$.
%	If $\frac{M}{N}$ is integer, the optimal choice for $\mathbf{m}$  is 
%	\begin{align}
%		m_i = \min \{ \frac{M}{N}, \mathcal{L}_i-N \}
%	\end{align}
%	If $\frac{M}{N}$ is not integer, let $t = M- N\lfloor \frac{M}{N} \rfloor$. 
	The optimal choice for $\mathbf{m}$ is 
	\begin{align}
		\forall i \in [t]:\ \  &m_i =  \min \{\lfloor \frac{M}{N} \rfloor+1, (\mathcal{L}_i-N)^+\}\\
		\forall i \in \{t+1,\dots,L\}:\ \ &m_i = \min \{\lfloor \frac{M}{N} \rfloor, (\mathcal{L}_i-N)^+\}
	\end{align}
\end{lemma}

%\begin{proof}
%	According to Lemma~\ref{lm:sizelargerthanN}, for $L>N$, the minimum number of required transmissions can be obtained the following optimization
%	\begin{align}
%		R^* %&= \min_{\mathcal{P}} \min_{\mathbf{m}} \sum_{i=1}^{L}R(|\wp_i|,m_i)\\
%		%& = \min_{\mathcal{P}:  |\wp_i| > N, \forall i\in [L]}  \min_{\mathbf{m}} \sum_{i=1}^{L}R(|\wp_i|,m_i)\\
%		%&=  \min_{\mathcal{P}:  |\wp_i| > N, \forall i\in [L]}  \min_{\mathbf{m}} \sum_{i=1}^{L}(|\wp_i| - m_i)\\
%		%& = \min_{\mathcal{P}:  |\wp_i| > N, \forall i\in [L]} \sum_{i=1}^{L}|\wp_i| - \max_{\mathbf{m}} \sum_{i=1}^{L}m_i\\
%		& = K-M-\max_{\mathcal{P}:  |\wp_i| > N, \forall i\in [L]}\max_{\mathbf{m}} \sum_{i=N+1}^{L}m_i
%	\end{align}
%	Hence, minimizing the total number of required transmission is equivalent to maximizing the number of side information messages which are used in coding subspace $\wp_{N+1},\dots,\wp_L$.
%    If $\frac{M}{N}$ is integer, the optimal choice is $m_1 = \dots=m_L = \frac{M}{N}$.
%    If $\frac{M}{N}$ is not integer, let $t = M- N\lfloor \frac{M}{N} \rfloor$. 
%    The optimal choice is $\forall i \in [t]$, $m_i = \lfloor \frac{M}{N} \rfloor +1$ and $\forall i \in \{t+1,\dots,L\}$, $m_i = N\lfloor \frac{M}{N} \rfloor$. 
%\end{proof}

\begin{lemma}\label{lm:optL}
	The optimal partitions, $\{\mathcal{L}_1,\dots,\mathcal{L}_L\}$ satisfy
	\begin{align}
		\forall i \in [t]:\ \  &\mathcal{L}_i =  \lfloor \frac{M}{N} \rfloor+N + 1\label{eq:L1}\\
		\forall i \in \{t+1,\dots,L-1\}:\ \ &\mathcal{L}_i =  \lfloor \frac{M}{N} \rfloor+N \label{eq:L2}\\
		\forall i = L:\ \   &\mathcal{L}_L = K- (L-1)(\lfloor \frac{M}{N} \rfloor +N)-t\label{eq:L3}
	\end{align}
	where $t = M- N\lfloor \frac{M}{N} \rfloor$. 
\end{lemma}
\begin{proof}
	According to Lemma~\ref{lm:vector4sideinformation},  in order to maximize the number of side information messages in each coding subspace, it is sufficiently optimal to have
	\begin{align}
		\forall i \in [t]:\ \  &\mathcal{L}_i =  \lfloor \frac{M}{N} \rfloor+1+N\\
		\forall i \in \{t+1,\dots,L-1\}:\ \ &\mathcal{L}_i =  \lfloor \frac{M}{N} \rfloor+N
	\end{align}
	Then, the total size of the first $L-1$ coding subspace is
	\begin{align}
		\sum_{i=1}^{L-1} \mathcal{L}_i = t(\lfloor \frac{M}{N} \rfloor+1+N) + (L-1-t)(\lfloor \frac{M}{N} \rfloor+N) = (L-1)(\lfloor \frac{M}{N} \rfloor +N)+t\label{eq:sumL}
	\end{align}
	Then the size of the last coding subspace can only be $q = K- (L-1)(\lfloor \frac{M}{N} \rfloor +N)-t$.
	If $q > \lfloor \frac{M}{N} \rfloor+N$, then we can further decompose the last coding subspace into two smaller coding subspaces with size $\lfloor \frac{M}{N} \rfloor+N$ and $q- \lfloor \frac{M}{N} \rfloor+N$.
	Thus, we can assume that $L$ is large enough such that $q < \lfloor \frac{M}{N} \rfloor+N$.
	If $N < q < \lfloor \frac{M}{N} \rfloor+N$, the number of side information messages in this coding subspace is equal to $q-N$.
	If $q\le N$, the number of transmissions required for the last coding subspace is equal to its size and hence the number of side information messages is zero.
%	As we have shown in the proof of Lemma~\ref{lm:minisizecodingsubspace}, it is possible to redistribute the $q$ messages to the first $q$ coding subspaces and the total number of transmissions remains the same.
\end{proof}

Now we are ready to prove Theorem~\ref{Thm:main}.
\begin{proof}(Theorem~\ref{Thm:main})
	%Let $t = M- N\lfloor \frac{M}{N} \rfloor$.
	As has been shown in the proof of Lemma~\ref{lm:optL}, Eqn.~\eqref{eq:sumL}, for the optimal partition, $\{\mathcal{L}_1,\dots,\mathcal{L}_L\}$, we have $\sum_{i=1}^{L-1} \mathcal{L}_i = (L-1)(\lfloor \frac{M}{N} \rfloor +N)+t$.
	Since the total number of messages is $K$, we have 
	\begin{align}
		K\ge  \sum_{i=1}^{L-1} \mathcal{L}_i = (L-1)(\lfloor \frac{M}{N} \rfloor +N)+t 
	\end{align}
	Additionally, $L$ is an integer, we have the optimal number of coding subspaces is
	\begin{align}
		L^* = \left\lceil\frac{K-t}{\lfloor \frac{M}{N} \rfloor +N}\right\rceil 
		% =  \left\lfloor\frac{K-M+ N\lfloor \frac{M}{N} \rfloor}{ \lfloor \frac{M}{N} \rfloor+N}\right\rfloor+1
		\label{eq:L*}
	\end{align}
	The side information vector $\mathbf{m} = \{m_1,\dots,m_L\}$ satisfies
	\begin{align}
	\forall i \in [t]:\ \  &m_i = \lfloor \frac{M}{N} \rfloor+1\\
	\forall i \in \{t+1,\dots,L^*-1\}:\ \ &m_i = \lfloor \frac{M}{N} \rfloor\\
	\forall i= L^*:\ \ &m_L = (K- (L^*-1)(\lfloor \frac{M}{N} \rfloor +N)-t -N)^+\label{eq:ml}
\end{align}
	Hence the total number of required transmissions is
%	\begin{align}
%		R^* &= \sum_{i= 1}^{L^*} R(\mathcal{L}_i,m_i) \\
%		&= (L^*-1)N + \min\{\mathcal{L}_L,N\}\\
%		& = N\left\lfloor\frac{K-M+ N\lfloor \frac{M}{N} \rfloor}{ \lfloor \frac{M}{N} \rfloor+N}\right\rfloor +
%		\min \left\{K- (L^*-1)(\lfloor \frac{M}{N} \rfloor +N)-t,N\right\}
%	\end{align}
	\begin{align}
		R^* &= K-\max_{\mathcal{L} \in \Pi(K)} \max_{\mathbf{m}} \sum_{i =1}^{|\mathcal{L}|}m_i\\
		& = K-M -\max_{\mathbf{m}}\sum_{i=N+1}^{L^*}m_i\\
		& = \left\{
		\begin{aligned}
		&K-M &\text{if } L^*\le N\\
		&K-M-((L^*-N)V)^+ &\text{if } L^*= N+1\\
		&K-M - (L^*-1-N)^{+}\bar{M}-((L^*-N)V)^+ &\text{if } L^*> N+1\\
		\end{aligned}
		\right.\\
		& = K-M - (L^*-1-N)^{+}\bar{M}- ((L^*-N)V)^+
		%& = K-M-(\hat{L} -N)^+\bar{M}-(K-\hat{L}(\bar{M}+N)-t-N)^+
		%& = K-M- \left( \left\lfloor\frac{K-M+ N\lfloor \frac{M}{N} \rfloor}{ \lfloor \frac{M}{N} \rfloor+N}\right\rfloor-N\right)^+\lfloor\frac{M}{N}\rfloor-(K-  \left\lfloor\frac{K-M+ N\lfloor \frac{M}{N} \rfloor}{ \lfloor \frac{M}{N} \rfloor+N}\right\rfloor(\lfloor \frac{M}{N} \rfloor +N)-t -N)^+
	\end{align}
	where $\bar{M} = \lfloor\frac{M}{N}\rfloor$, $t = M- N\bar{M}$, $L^* = \left\lceil\frac{K-t}{ \bar{M}+N}\right\rceil$ and $V = (K- (L^*-1)(\bar{M} +N)-t -N)/(L^*-N)$.

\end{proof}

\begin{example*}
	For Example~\ref{eg:1}, we have $K=13$, $M=5$ and $N=2$. According to Theorem~\ref{Thm:main}, it is easy to get that $\bar{M} = \lfloor\frac{M}{N}\rfloor = 2$, $t = M- N\bar{M}= 1$ and $L^* = \left\lceil\frac{K-t}{ \bar{M}+N}\right\rceil =3$. The minimum number of required transmissions is $R^* = 6$.
\end{example*}

\begin{theorem}
	 For given total number of messages $K$ and number of side information messages $M$, it is optimal to download $K-M$ transmissions by using the trivial MDS coding scheme with all messages in one coding subspace if either of the following conditions is satisfied.
	 \begin{enumerate}
	 	\item $N> M$.
	 	\item $N^2+N \ge K-M$
	 	%$L^*  = \left\lceil\frac{K-M+N\lfloor\frac{M}{N}\rfloor}{\lfloor \frac{M}{N} \rfloor +N}\right\rceil \le N$
	 	%\item $L^* = N+1$ and $m_{L^*} =(K- (L^*-1)(\lfloor \frac{M}{N} \rfloor +N)-t -N)^+=0 $
	 \end{enumerate}
\end{theorem}
\begin{proof}
	If the first condition is satisfied, $N>M$, as it is optimal to only consider the coding schemes with partition of size larger than $N$, it can be shown that the number of side information messages used in the $N$-th coding subspace is $m_N = 0$.
	Hence, we have $m_{N+1}=\dots,m_{L} =0$. The minimum number of required transmission is $K-M$.
	
	If the second condition holds, $N^2+N \ge K-M$. Let us further assume that $N^2>K-M$, then we have
	\begin{align}
		&N^2+N\lfloor\frac{M}{N}\rfloor > K-M+ N\lfloor\frac{M}{N}\rfloor\\
		\Leftrightarrow& N+1> \frac{K-M+ N\lfloor\frac{M}{N}\rfloor}{N+\lfloor\frac{M}{N}\rfloor}+1\ge L^*
	\end{align}
	Hence we have $L^* \le N$. According to Lemma~\ref{lm:sizelessthanN}, the minimum number of required transmissions is $K-M$.
	If $N^2\le K-M\le N^2+N$, then it can be shown that $L^*=N+1$ and 
	\begin{align}
		& K-M - N\lfloor\frac{M}{N}\rfloor   \le N^2+N -N\lfloor\frac{M}{N}\rfloor \\
		\Leftrightarrow& K-N(\lfloor\frac{M}{N}\rfloor  +N) - (M-N\lfloor\frac{M}{N}\rfloor )-N \le 0\\
		\Leftrightarrow& m_{N+1} = (K-N(\lfloor\frac{M}{N}\rfloor  +N) - (M-N\lfloor\frac{M}{N}\rfloor )-N)^+=0
	\end{align}
	Since $L^* = N+1$ and $m_{N+1}=0$, according to Theorem~\ref{Thm:main}, $R^* =K-M$.
	Thus, if any of the three conditions is satisfied, it is optimal to use the trivial MDS coding scheme with $K-M$ transmissions which takes all messages in one coding subspace.
\end{proof}

\section{Optimal Linear Coding Scheme}\label{sec:CodingScheme}
In this section, we show how to construct an optimal coding scheme for the single-server multi-message private information retrieval with side information problem.

For any single-server multi-message private information retrieval with side information problem with 
$K$ total messages, $M$ side information messages and $N$ demand messages, we first compute $L^*$ defined by Eqn.~\eqref{eq:L*} and $m_L$ defined by Eqn.~\eqref{eq:ml}.
If $L^*\le N+1$ and $m_L = 0$, then $R^* = K-M$. It is trivial that the optimal linear coding scheme is the MDS coding scheme that takes all messages into one coding subspace.
If $L^*> N+1$ or $L^*=N+1$, $m_L > 0$, we use the following  steps to construct the optimal linear coding schemes:
\begin{itemize}
	\item[] \textbf{Step 1}: 
	The user creates a set of $L^*$ subsets, denoted by $\{\wp_1,\dots,\wp_{L^*}\}$ and $\forall i \in [L^*]$, the size of $\wp_i$ satisfies:
	\begin{align}
		|\wp_i| = \left\{
		\begin{aligned}
			&\lfloor\frac{M}{N}\rfloor + N+1, & \forall t \in \{1,\dots,t\} \\
			&\lfloor\frac{M}{N}\rfloor + N, & \forall  t \in \{t+1,\dots,L^*-1\}\\
			& K- (L-1)(\lfloor \frac{M}{N} \rfloor +N)-t, & \text{for}\ i =L^*
		\end{aligned}
		\right.
	\end{align}
	where $t = M-N\lfloor \frac{M}{N} \rfloor$. Let $c_i$ for $i \in [L^*]$ denote the number of demand messages in subset $\wp_i$ and initiate to $0$

	\item[] \textbf{Step 2}: 
	For the first demand message $X_{W_1}$, the user randomly selects one subset $\wp_i$ ($i \in [L^*]$) to contain it with probability $\frac{|\wp_i|}{K}$.
	The user updates $c_i = c_i+1$.
	Then for the $j$-th demand message ($j\in[N]$), the user randomly selects one subset $\wp_u$ ($u \in [L^*]$) to contain it with probability $\frac{|\wp_u|-c_u}{K-j+1}$.
	Iteratively, the user places all demand messages into the subsets.

	\item[] \textbf{Step 3}: For each subset $\wp_i$ with $c_i > 0$, the user randomly selects $m_i$ side information messages to put into $\wp_i$, where $m_i$ satisfies:
	\begin{align}
		m_i = \left\{
		\begin{aligned}
		   &\lfloor\frac{M}{N} \rfloor +1, & 1\le i\le t\\
		   &\lfloor \frac{M}{N} \rfloor, & t+1 \le i \le L^*-1\\
		   &(|\wp_{L^*}|-N)^+, & i=L^*
		\end{aligned}
		\right.
	\end{align}
%	For each selected subset with index  between $1$ and $t$, the user put the demand messages and $\lfloor \frac{M}{N} \rfloor +1$ side information message in it.
%	For each selected subset with index between $t+1$ and $L^*$,  the user put the demand messages and $\lfloor \frac{M}{N} \rfloor$ side information message in it.
	
	\item[] \textbf{Step 4}: The user randomly distributes the other messages to fill up the  remaining empty spaces in each subset.
	
	\item[] \textbf{Step 5}: The user sends queries to the server according to the coding scheme which satisfies the MDS-Condition in each coding subspace $\wp_i$ ($\forall i \in [L^*]$) with $R(|\wp_i|,m_i)$ transmissions. 
\end{itemize}

We name the coding schemes constructed by this method as Partition-and-MDS-Coding scheme, which is a modified based on optimal coding scheme for single demand message~\cite{2017arXiv170900112K}. The way we select subsets for demand messages is related to the URN problem. The probability of any $N$ messages to be the demand messages follows the binomial distribution.

\begin{theorem}
	The Partition-and-MDS-Coding schemes satisfies the Retrieval Condition and the Privacy Condition.
\end{theorem}
\begin{proof}
	For each coding subspace $\wp_i$, if it contains demand messages,
	the number of transmissions $R(|\wp_i|,m_i)$ and the number of side information messages $m_i$ in such coding subspace satisfy $R(|\wp_i|,m_i)+m_i = |\wp_i|$.
	Additionally, the Partition-and-MDS-Coding scheme satisfies MDS-Condition in very coding subspace.
	Thus, all missing messages in $\wp_i$ can be successfully decoded, including the demand messages.
	Therefore, the Retrival Condition is satisfied.
	
	The probability that any $N$ messages (e.g. $\{X_{Z_1},\dots,X_{Z_N}\}$) are the demand messages can be computed as
	\begin{align}
		\Pr(\mathcal{W} = \{Z_1,\dots,Z_N\}) &= N!\, \Pr(W_1=Z_1,W_2=Z_2,\dots,W_N=Z_N)\\
		& = N!\, \Pr(W_1=Z_1)Pr(W_2=Z_2,\dots,W_N=Z_N|W_1=Z_1)\\
		& = N!\, \prod_{i=1}^{N} \Pr(W_i = Z_i|W_1^{i-1} = Z_1^{i-1})
	\end{align}
	According to the construction of the Partition-and-MDS-Coding scheme and assume that $Z_i \in \wp_j$, we have
	\begin{align}
		\Pr(W_i = Z_i|W_1^{i-1} = Z_1^{i-1}) &= \Pr(W_i \in \wp_j| W_1^{i-1} = Z_1^{i-1}) \Pr(W_i=Z_i| W_i \in \wp_j, W_1^{i-1} = Z_1^{i-1} )\\
		& = \frac{|\wp_j| - \wp_j \setminus (\wp_j \cap \{Z_1,\dots,Z_{i-1}\})}{K-i+1}\frac{1}{|\wp_j| - \wp_j \setminus (\wp_j \cap \{Z_1,\dots,Z_{i-1}\})}\\
		& = \frac{1}{K-i+1}
	\end{align}
	Hence, we have
	\begin{align}
		\Pr(\mathcal{W} = \{Z_1,\dots,Z_N\}) &= N!\, \prod_{i=1}^{N} \frac{1}{K-i+1} = \frac{N!}{K(K-1)\cdots (K-N+1)}
		= \frac{1}{\binom{K}{N}}
	\end{align}
	Since there are $\binom{K}{N}$ possible demand message pairs with size $N$, every $N$-message pair is equally likely to be the demand messages, which satisfies the Privacy Condition of multi-message PIR.
\end{proof}

\begin{example*}
	We construct the linear coding scheme in Example~\ref{eg:1} by using the Partition-and-MDS-Coding method. It is easy to compute and verify that $L^* = 3>N$ and $m_3 = 2 >0$.
	Step 1: We first create three coding subspace ($\wp_1$, $\wp_2$, $\wp_3$) with size $|\wp_1| = 5$, $|\wp_2| = 4$ and $|\wp_3|=4$.
	Step 2: We randomly select one coding subspace from $\{\wp_1,\wp_2,\wp_3\}$ to contain the first demand message $X_2$ with probability $\frac{5}{13}$, $\frac{4}{13}$ and $\frac{4}{13}$, respectively.
	Suppose $\wp_1$ is chosen. 
	Then for the second demand message $X_2$, we randomly select one coding subspace from $\{\wp_1,\wp_2,\wp_3\}$ to contain the first demand message $X_2$ with probability $\frac{4}{13}$, $\frac{4}{13}$ and $\frac{4}{13}$, respectively.
	Suppose $\wp_3$ is chosen.
	Step 3: For $\wp_1$ and $\wp_3$, the coding subspaces which are chosen to contain demand messages, we randomly distribute $3$ and $2$ side information messages into them, respectively.
	Suppose $X_1,X_4,X_6$ are placed in $\wp_1$ and $X_7,X_9$ are placed in $\wp_2$.
	Step 4: Randomly distribute the remaining messages into the coding subspaces. 
	Suppose we get $\wp_1 = \{X_1,X_2,X_4,X_6,X_8\}$, $\wp_2 = \{X_3,X_{10},X_{11},X_{13}\}$ and $\wp_3 = \{X_5,X_7,X_9,X_{12}\}$.
	Step 5: The user generates queries according to the linear coding scheme $\mathbf{T} = \{T_1,\dots,T_6\}$ shown in Example~\ref{eg:1}. 
	
	From the server's perspective, the probability for any two messages to be the demand message is the same, which is $\frac{1}{\binom{13}{2}} =\frac{1}{78}$. Thus, the server cannot infer any information about the indices of the demand messages.
\end{example*}

\section*{appendix}
We provide an alternative proof for the converse of the minimum number of required transmissions for the single-server multi-message private information retrieval with side information problem.
The proof techniques are inspired by~\cite{2017arXiv170903022C}. 
Suppose each message has $L$ bits and messages are independent from each other, i.e.,
\begin{align}
	H(X_1,\dots,X_K) &= H(X_1)+ \dots + H(X_K),\\
	H(X_1) &= \dots = H(X_K) = L.
\end{align}
Recall that $\mathbf{W}$ and $\mathbf{S}$ denote the sets of indices of demand messages and side information messages, respectively.
Let $Q^{[\mathbf{W},\mathbf{S}]}$ denote the query that is generated for side information indexed by $\mathbf{S}$ and demand messages indexed by $\mathbf{W}$.
Let $A^{[\mathbf{W},\mathbf{S}]}$ denote the answer generated by the server after receiving query $Q^{[\mathbf{W},\mathbf{S}]}$.
Since the answer generated by the server is a deterministic function of the query and messages, we have
\begin{align}
	H(A^{[\mathbf{W},\mathbf{S}]}|Q^{[\mathbf{W},\mathbf{S}]},X_{1:K}) = 0\label{eq:deta}
\end{align}
The retrieval condition~\eqref{eq:retcon} is equivalent to
\begin{align}
	H(X_{\mathbf{W}}|A^{[\mathbf{W},\mathbf{S}]},Q^{[\mathbf{W},\mathbf{S}]},X_{\mathbf{S}}) = 0\label{eq:rcap}.
\end{align}
The privacy condition~\eqref{eq:pricon} is equivalent to the condition that for any $\mathbf{W},\mathbf{W}' \subseteq [K]$ and $|\mathbf{W}|=|\mathbf{W}'| = N$, there exists $\mathbf{S} \subseteq [K]\setminus \mathbf{W}$, $\mathbf{S}' \subseteq [K]\setminus \mathbf{W}'$ and $|\mathbf{S}| = |\mathbf{S}'| = M$ such that
\begin{align}
(A^{[\mathbf{W},\mathbf{S}]},Q^{[\mathbf{W},\mathbf{S}]},X_{1:K})\sim (A^{[\mathbf{W}',\mathbf{S}']},Q^{[\mathbf{W}',\mathbf{S}']},X_{1:K})\label{eq:w}.
\end{align}
where $A\sim B$ means that $A$ and $B$ are identically distributed.

Suppose $\mathbf{W}_0 = \mathbf{W}$ and $\mathbf{S}_0 = \mathbf{S}$ are the set of indices of demand messages and side information messages, respectively.

Then the total number of download bits ($D$) can be lower-bounded as follows.
\begin{align}
	D&\ge H(A^{[\mathbf{W}_0,\mathbf{S}_0]}|Q^{[\mathbf{W}_0,\mathbf{S}_0]},X_{\mathbf{S}_0})\\
	& = H(X_{\mathbf{W}_0},A^{[\mathbf{W}_0,\mathbf{S}_0]}|Q^{[\mathbf{W}_0,\mathbf{S}_0]},X_{\mathbf{S}_0})- H(X_{\mathbf{W}_0}|A^{[\mathbf{W}_0,\mathbf{S}_0]},Q^{[\mathbf{W}_0,\mathbf{S}_0]},X_{\mathbf{S}_0})\\
	& \stackrel{\eqref{eq:correctness}}{=} H(X_{\mathbf{W}_0}|Q^{[\mathbf{W}_0,\mathbf{S}_0]},X_{\mathbf{S}_0}) + H(A^{[\mathbf{W}_0,\mathbf{S}_0]}|Q^{[\mathbf{W}_0,\mathbf{S}_0]},X_{\mathbf{W}_0\cup\mathbf{S}_0})\\
	&= NL+H(A^{[\mathbf{W}_0,\mathbf{S}_0]}|Q^{[\mathbf{W}_0,\mathbf{S}_0]},X_{\mathbf{W}_0\cup\mathbf{S}_0}).
\end{align}

According to the privacy condition, for any $\mathbf{W}_i \subseteq [K]$, $|\mathbf{W}_i|= N$, there exists $\mathbf{S}_i \subseteq [K]\setminus \mathbf{W}_i$ which satisfies the retrieval condition, i.e.,
\begin{align}
H(X_{\mathbf{W}_j}| A^{[\mathbf{W}_j,\mathbf{S}_j]},Q^{[\mathbf{W}_j,\mathbf{S}_j]},X_{\mathbf{S}_j}) = 0\label{eq:correctness}.
\end{align}
We have
\begin{align}
	H(A^{[\mathbf{W}_0,\mathbf{S}_0]}|Q^{[\mathbf{W}_0,\mathbf{S}_0]},X_{[K]}) = \min_{\mathbf{S}_i}
	H(A^{[\mathbf{W}_i,\mathbf{S}_i]}|Q^{[\mathbf{W}_i,\mathbf{S}_i]},X_{[K]}).\label{eq:st}
\end{align}
Hence, for the special case $i =1$, we have
\begin{align}
	D &\ge NL +\min_{\mathbf{S}_1}
	H(A^{[\mathbf{W}_1,\mathbf{S}_1]}|Q^{[\mathbf{W}_1,\mathbf{S}_1]},X_{\mathbf{W}_0\cup\mathbf{S}_0})\\
	& =NL+ \min_{\mathbf{S}_1} H(X_{\mathbf{W}_1,\mathbf{S}_1}|Q^{[\mathbf{W}_1,\mathbf{S}_1]},X_{\mathbf{W}_0\cup \mathbf{S}_0}) + H(A^{[\mathbf{W}_1,\mathbf{S}_1]}|Q^{[\mathbf{W}_1,\mathbf{S}_1]},X_{\mathbf{W}_0^1\cup \mathbf{S}_0^1})\nonumber\\ 
	&\ \ \ -H(X_{\mathbf{W}_1\cup\mathbf{S}_1}|A^{[\mathbf{W}_1,\mathbf{S}_1]},Q^{[\mathbf{W}_1,\mathbf{S}_1]},X_{\mathbf{W}_0\cup\mathbf{S}_0})\\
	& =NL+\min_{\mathbf{S}_1} H(X_{\mathbf{W}_1,\mathbf{S}_1}|Q^{[\mathbf{W}_1,\mathbf{S}_1]},X_{\mathbf{W}_0\cup \mathbf{S}_0})+ H(A^{[\mathbf{W}_1,\mathbf{S}_1]}|Q^{[\mathbf{W}_1,\mathbf{S}_1]},X_{\mathbf{W}_0^1\cup \mathbf{S}_0^1}) \nonumber\\
	&\ \ \ - H(X_{\mathbf{S}_1}|A^{[\mathbf{W}_1,\mathbf{S}_1]},Q^{[\mathbf{W}_1,\mathbf{S}_1]},X_{\mathbf{W}_0\cup\mathbf{S}_0}) - H(X_{\mathbf{W}_1}|A^{[\mathbf{W}_1,\mathbf{S}_1]},Q^{[\mathbf{W}_1,\mathbf{S}_1]},X_{\mathbf{W}_0\cup\mathbf{S}_0^1})\\
	& \stackrel{\eqref{eq:correctness}}{=}NL+\min_{\mathbf{S}_1} H(X_{\mathbf{W}_1,\mathbf{S}_1}|Q^{[\mathbf{W}_1,\mathbf{S}_1]},X_{\mathbf{W}_0\cup \mathbf{S}_0})- H(X_{\mathbf{S}_1}|A^{[\mathbf{W}_1,\mathbf{S}_1]},Q^{[\mathbf{W}_1,\mathbf{S}_1]},X_{\mathbf{W}_0\cup\mathbf{S}_0})\nonumber \\
	&\ \ \ 
	+H(A^{[\mathbf{W}_1,\mathbf{S}_1]}|Q^{[\mathbf{W}_1,\mathbf{S}_1]},X_{\mathbf{W}_0^1\cup \mathbf{S}_0^1}) 
\end{align}
We can apply the following substitutions iteratively
\begin{align}
	\min_{\mathbf{S}_i} H(A^{[\mathbf{W}_i,\mathbf{S}_i]}|Q^{[\mathbf{W}_i,\mathbf{S}_i]},X_{\mathbf{W}_0^i\cup \mathbf{S}_0^i}) = \min_{\mathbf{S}_{i+1}} H(A^{[\mathbf{W}_{i+1},\mathbf{S}_{i+1}]}|Q^{[\mathbf{W}_{i+1},\mathbf{S}_{i+1}]},X_{\mathbf{W}_0^{i+1}\cup \mathbf{S}_0^{i+1}}).
\end{align}
Suppose after $T$ substitutions, we have 
 \begin{align}
 \mathbf{W}_0^T \cup \mathbf{S}_0^T = [K] \label{eq:k}.
 \end{align}

Then we have the lower-bound for $D$ as follows.
\begin{align}
	D &\ge NL + \min_{\mathbf{S}_1,\dots,\mathbf{S}_T} \sum_{i=1}^{T}\left[ H(X_{\mathbf{W}_i,\mathbf{S}_i}|Q^{[\mathbf{W}_i,\mathbf{S}_i]},X_{\mathbf{W}_0^{i-1}\cup \mathbf{S}_0^{i-1}})- H(X_{\mathbf{S}_i}|A^{[\mathbf{W}_i,\mathbf{S}_i]},Q^{[\mathbf{W}_i,\mathbf{S}_i]},X_{\mathbf{W}_0^{i-1}\cup\mathbf{S}_0^{i-1}})\right]\nonumber \\
	&\ \ \ 
	+H(A^{[\mathbf{W}_T,\mathbf{S}_T]}|Q^{[\mathbf{W}_T,\mathbf{S}_T]},X_{\mathbf{W}_0^T\cup \mathbf{S}_0^T})\\
	&\stackrel{\eqref{eq:deta}\eqref{eq:k}}{=} NL + \min_{\mathbf{S}_1,\dots,\mathbf{S}_T} \sum_{i=1}^{T} \left[H(X_{\mathbf{W}_i,\mathbf{S}_i}|Q^{[\mathbf{W}_i,\mathbf{S}_i]},X_{\mathbf{W}_0^{i-1}\cup \mathbf{S}_0^{i-1}})- H(X_{\mathbf{S}_i}|A^{[\mathbf{W}_i,\mathbf{S}_i]},Q^{[\mathbf{W}_i,\mathbf{S}_i]},X_{\mathbf{W}_0^{i-1}\cup\mathbf{S}_0^{i-1}})\right]\label{eq:lb1}	
\end{align}
Note that each term in the summation is non-negative, since
\begin{align}
H(X_{\mathbf{W}_i,\mathbf{S}_i}|Q^{[\mathbf{W}_i,\mathbf{S}_i]},X_{\mathbf{W}_0^{i-1}\cup \mathbf{S}_0^{i-1}}) &\ge H(X_{\mathbf{S}_i}|Q^{[\mathbf{W}_i,\mathbf{S}_i]},X_{\mathbf{W}_0^{i-1}\cup \mathbf{S}_0^{i-1}})\\
&\ge H(X_{\mathbf{S}_i}|A^{[\mathbf{W}_i,\mathbf{S}_i]},Q^{[\mathbf{W}_i,\mathbf{S}_i]},X_{\mathbf{W}_0^{i-1}\cup\mathbf{S}_0^{i-1}})
\end{align}

In order to get a lower-bound for the total number of download bits ($D$), we need to minimize the summation.
And this lower-bound works for any choice of $\{\mathbf{W}_1,\dots,\mathbf{W}_T\}$.
We will construct a special set $\{\mathbf{W}_1,\dots,\mathbf{W}_T\}$ such that we can compute the minimum of the summation.

For any $i \in \mathbf{W}_0$, let $\mathbf{V}_i \subset \mathbf{S}_0$ denote the minimum subset such that
\begin{align}
H(X_i|A^{[\mathbf{W}_0],\mathbf{S}_0},Q^{[\mathbf{W}_0],\mathbf{S}_0},X_{\mathbf{V}_i}) = 0\label{eq:v}
\end{align}
Without loss of optimality, we may assume that $\cup_{i\in \mathbf{W}_0} \mathbf{V}_i = \mathbf{S}_0$. 
Let $i^* = \arg\max_{i} |\cup_{j\in \mathbf{W}_0\setminus i}\mathbf{V}_j|$.
We construct $W_t$ for $t \in [T]$ as the following steps.
\begin{enumerate}
	\item Put indices $\mathbf{W}_0 \setminus i^*$ into $\mathbf{W}_t$.
	\item Add another index $i_t$ into $\mathbf{W}_t$, where $i_t \not \in \mathbf{W}_0^{t-1}$.
\end{enumerate}
After we have $\mathbf{W}_t$, we can select $\mathbf{S}_t$ to maximize the corresponding term in the summation in Equation~\eqref{eq:lb1}.
In such way, each round we add a new index in $\mathbf{W}_0^{t}$. Hence, after $T = K-N$ rounds, we have $\mathbf{W}_0^T = [K]$.

When the newly added index $i_t \in \mathbf{S}_0^{t-1}$, the optimal choice for $\mathbf{S}_t$ is $\mathbf{S}_t \subset (\mathbf{W}_0^{t-1}\cup\mathbf{S}_0^{t-1}\setminus i_t)$. 
In such case 
\begin{align}
	H(X_{\mathbf{W}_t,\mathbf{S}_t}|Q^{[\mathbf{W}_t,\mathbf{S}_t]},X_{\mathbf{W}_0^{t-1}\cup \mathbf{S}_0^{t-1}})- H(X_{\mathbf{S}_t}|A^{[\mathbf{W}_t,\mathbf{S}_t]},Q^{[\mathbf{W}_t,\mathbf{S}_t]},X_{\mathbf{W}_0^{t-1}\cup\mathbf{S}_0^{t-1}})=0,
\end{align} 
implying that this choice achieves the minimum. 

By assumption, $A^{[\mathbf{W}_t,\mathbf{S}_t]}$ given side information $X_{\mathbf{S}_t}$ permits to decode $X_{\mathbf{W}_t}.$
It is possible that the same $A^{[\mathbf{W}_t,\mathbf{S}_t]},$ given the same side information $X_{\mathbf{S}_t},$ {\it also}  permits to decode further messages.
Let us denote the indices of these decodable messages by ${\mathbf{U}_t}$ (noting that ${\mathbf{U}_t}$ may be the empty set), and the corresponding messages by $X_{\mathbf{U}_t}.$ Clearly, $\mathbf{U}_t \subseteq [K] \setminus (\mathbf{W}_t\cup \mathbf{S}_t),$
and the definition of $X_{\mathbf{U}_t}$ can be written as
%%% PREVIOUS TEXT:
%Let $X_{\mathbf{U}_t}$ be the messages which can also be decoded from coding scheme $A^{[\mathbf{W}_t,\mathbf{S}_t]}$ given side information $X_{\mathbf{S}_t}$ (besides $X_{\mathbf{W}_t}$), where $\mathbf{U}_t \subseteq [K] \setminus (\mathbf{W}_t\cup \mathbf{S}_t)$, i.e.,
\begin{align}
H(X_{\mathbf{U}_t}| A^{[\mathbf{W}_t,\mathbf{S}_t]},Q^{[\mathbf{W}_t,\mathbf{S}_t]},X_{\mathbf{S}_t}) = 0.
\end{align}

Similarly, when the newly added index $i_t \in \mathbf{U}_0^t$, we can show that the optimal choice for $\mathbf{S}_t$ is $\mathbf{S}_t \subset (\mathbf{W}_0^{t-1}\cup\mathbf{S}_0^{t-1}\setminus i_t)$.
In such cases, 
\begin{align}
	H(X_{\mathbf{W}_t,\mathbf{S}_t}|Q^{[\mathbf{W}_t,\mathbf{S}_t]},X_{\mathbf{W}_0^{t-1}\cup \mathbf{S}_0^{t-1}})- H(X_{\mathbf{S}_t}|A^{[\mathbf{W}_t,\mathbf{S}_t]},Q^{[\mathbf{W}_t,\mathbf{S}_t]},X_{\mathbf{W}_0^{t-1}\cup\mathbf{S}_0^{t-1}})=H(X_{i_t}) = L.
\end{align}
which achieves the minimum.

Now, the difficulty is minimizing those terms in the summation of Equation~\eqref{eq:lb1} where $i_t \not \in (\mathbf{W}_0^{t-1}\cup \mathbf{S}_0^{t-1}\cup \mathbf{U}_0^{t-1})$.
To deal with them, we need to further exploit the lower-bound expression.
Since $X_{\mathbf{W}_t\cup\mathbf{S}_t}$ is independent from the query $Q^{[\mathbf{W}_t,\mathbf{S}_t]}$, we have
\begin{align}
	\sum_{i=1}^{T} H(X_{\mathbf{W}_i,\mathbf{S}_i}|Q^{[\mathbf{W}_i,\mathbf{S}_i]},X_{\mathbf{W}_0^{i-1}\cup \mathbf{S}_0^{i-1}}) &= \sum_{i=1}^{T} H(X_{\mathbf{W}_i,\mathbf{S}_i}|X_{\mathbf{W}_0^{i-1}\cup \mathbf{S}_0^{i-1}}) \\
	&= (|\mathbf{W}_0^T\cup\mathbf{S}_0^T| - |\mathbf{W}_0\cup\mathbf{S}_0|)L \\
	&= (K-N-M)L.
\end{align}
Thus, the total number of download bits $D$ is also lower-bounded by
\begin{align}
	D &\ge  NL +  (|\mathbf{W}_0^T\cup\mathbf{S}_0^i|-|\mathbf{W}_0\cup\mathbf{S}_0|)L-\max_{\mathbf{S}_1,\dots,\mathbf{S}_T}\sum_{i=1}^{T} H(X_{\mathbf{S}_i}|A^{[\mathbf{W}_i,\mathbf{S}_i]},Q^{[\mathbf{W}_i,\mathbf{S}_i]},X_{\mathbf{W}_0^{i-1}\cup\mathbf{S}_0^{i-1}})\\
	& = (K-M)L - \max_{\mathbf{S}_1,\dots,\mathbf{S}_T}\sum_{i=1}^{T} H(X_{\mathbf{S}_i}|A^{[\mathbf{W}_i,\mathbf{S}_i]},Q^{[\mathbf{W}_i,\mathbf{S}_i]},X_{\mathbf{W}_0^{i-1}\cup\mathbf{S}_0^{i-1}})\label{eq:lb2}
\end{align}
Thus, we can also maximize the summation of conditional entropies in Equation~\eqref{eq:lb2} to get the lower-bound.
As we have shown before, for $\mathbf{W}_t$ with $i_t \in \{\mathbf{W}_0^{t-1}\cup\mathbf{S}_0^{t-1}\cup\mathbf{U}_0^{t-1}\}$, the optimal choice $\mathbf{S}_t \subset (\mathbf{W}_0^{t-1}\cup\mathbf{S}_0^{t-1}\setminus i_t)$, which implies
\begin{align}
	H(X_{\mathbf{S}_t}|A^{[\mathbf{W}_t,\mathbf{S}_t]},Q^{[\mathbf{W}_t,\mathbf{S}_t]},X_{\mathbf{W}_0^{t-1}\cup\mathbf{S}_0^{t-1}}) = 0.
\end{align} 
For $\mathbf{W}_t$ with $i_t \not\in \{\mathbf{W}_0^{t-1}\cup\mathbf{S}_0^{t-1}\cup\mathbf{U}_0^{t-1}\}$, since $\mathbf{W}_0\setminus i^* \subset \mathbf{W}_t$, we have $\cup_{j\in \mathbf{W}_0\setminus i^*} \mathbf{V}_j \subseteq \mathbf{S}_t$ to guarantee the decoding correctness of $X_{\mathbf{W}_0\setminus i^*}$.
Thus, we can upper-bound the corresponding conditional entropy by
\begin{align}
H(X_{\mathbf{S}_{t}}|A^{[\mathbf{W}_{t},\mathbf{S}_{t}]},Q^{[\mathbf{W}_{t},\mathbf{S}_{t}]},X_{\mathbf{W}_0^{t-1}\cup\mathbf{S}_0^{t-1}}) &\le H(X_{\mathbf{S}_t}|\mathbf{W}_0^{t-1}\cup\mathbf{S}_0^{t-1}\cup\mathbf{U}_0^{t-1})\\
& \le (|\mathbf{S}_t| - |\cup_{j\in \mathbf{W}_0\setminus i^*} \mathbf{V}_j|)
\end{align}
By assumption, $\cup_{j\in \mathbf{W}_0} \mathbf{V}_j = \mathbf{S}_0$, we have
\begin{align}
\max_{\mathbf{V}_1^N :|\cup_{i\in \mathbf{W}_0} \mathbf{V}_i| = M } |\mathbf{S}_t| -|\cup_{j\in \mathbf{W}_0\setminus i^*} \mathbf{V}_j| &= 
M-\min_{\mathbf{V}_1^N :|\cup_{i\in \mathbf{W}_0} \mathbf{V}_i| = M } \max_i|\cup_{j\in \mathbf{W}_0\setminus i} \mathbf{V}_j| \\
& 
\left\{
\begin{aligned}
&\le \left\lfloor \frac{M}{N} \right\rfloor & \text{ if } M\ge N\\
&= 0 &\text{ if } M< N
\end{aligned}
\right.\label{eq:mh}
\end{align} 
where the maximum is achieved when $\mathbf{V}_i \cap \mathbf{V}_j = \emptyset$ and ($M-N\left\lfloor \frac{M}{N} \right\rfloor$) $|\mathbf{V}_i|$'s are equal to $\left\lceil \frac{M}{N} \right\rceil$ and others are equal to $\left\lfloor \frac{M}{N} \right\rfloor$.

Therefore, for any $t \in [T]$, if $M<N$
\begin{align}
	H(X_{\mathbf{S}_{t}}|A^{[\mathbf{W}_{t},\mathbf{S}_{t}]},Q^{[\mathbf{W}_{t},\mathbf{S}_{t}]},X_{\mathbf{W}_0^{t-1}\cup\mathbf{S}_0^{t-1}}) = 0.
\end{align}
Otherwise, if $M \ge N$
\begin{align}
H(X_{\mathbf{S}_{t}}|A^{[\mathbf{W}_{t},\mathbf{S}_{t}]},Q^{[\mathbf{W}_{t},\mathbf{S}_{t}]},X_{\mathbf{W}_0^{t-1}\cup\mathbf{S}_0^{t-1}}) \le \left\lfloor \frac{M}{N} \right\rfloor.
\end{align}

\begin{lemma}
	If the number of side information messages is smaller than the number of demand messages, i.e. $M< N$, the minimum number of required transmissions is $K-M$.
\end{lemma}
\begin{proof}
	Suppose $M<N$, from Equation~\eqref{eq:mh} we have that
	\begin{align}
	M-\min_{\mathbf{V}_1^N :|\cup_{i\in \mathbf{W}_0} \mathbf{V}_i| = M } \max_i|\cup_{j\in \mathbf{W}_0\setminus i} \mathbf{V}_j| = 0.
	\end{align}
	Thus, each conditional entropy in the summation is zero, except the first term $H(X_{\mathbf{S}_0}) = ML$.
	Hence, we have $D\ge (K-M)L$ which gives $R = \frac{D}{L} \ge K-M$.
	Additionally, we know that the MDS coding scheme with $K-M$ is always a PIR scheme.
	Therefore, $R^* = K-M$.
\end{proof}

Based on this, we can conclude the following useful proportions.

\begin{proportion}
	If there exists $j \ne i$ ($i,j \in \mathbf{W}_0$) such that $\mathbf{V}_j = \mathbf{V}_i$, where $\mathbf{V}_i$ and $\mathbf{V}_j$ are defined in Equation~\eqref{eq:v}, then the minimum number of required transmissions is $K-M$.
\end{proportion}
\begin{proof}
	If $\mathbf{V}_i = \mathbf{V}_j$ ($i\ne j$), then when we construct $\mathbf{W}_t$, we can remove either $i$ or $j$ from $\mathbf{W}_0$ and keep the others.
	In such way, no matter what new index we add into $\mathbf{W}_t$, we have $H(X_{\mathbf{S}_t}|X_{\mathbf{S}_0}) = 0$. Hence, the lower bound for the number of transmissions is $K-M$ and we know that MDS coding scheme can achieve this lower bound.
\end{proof}

\begin{proportion}\label{pp:NumOfMes}
	For any PIR coding scheme, for each $\mathbf{V}_i \ne \emptyset$, defined by Equation~\eqref{eq:v}, given $X_{\mathbf{V}_i}$, besides decoding $X_i$, there must exist at least another $N-1$ messages that can also be decoded.
\end{proportion}
\begin{proof}
	Let $\mathbf{Y}_i$ denote the set of indices of the messages that can be decoded given $X_{\mathbf{V}_i}$.
	Apparently, $ i \in \mathbf{Y}_i$, since $X_i$ can be decoded.
	Suppose $|\mathbf{Y}_i| \le N-1$.
	Then for any $j \in \mathbf{V}_i$, the set of messages $X_{\mathbf{Y}_i\cup\{j\}}$ cannot be decoded given any $N$ messages in $[K]\setminus (\mathbf{Y}_i\cup\{j\})$.
	This is because  $\mathbf{V}_i$ are the minimum set of messages that are required to decode $X_i$.
	Hence,  $X_{\mathbf{Y}_i\cup\{j\}}$ cannot be the demand messages, which violates the privacy condition.
	Therefore, $|\mathbf{Y}_i| \ge N$.
\end{proof}

\begin{proportion}\label{pp:Y}
	For $i \in \mathbf{W}_0$, let $\mathbf{Y}_i$ denote the set of indices of the messages that can be decoded given $X_{\mathbf{V}_i}$.
	Without loss of optimality, we can assume that $\mathbf{Y}_i \cap \mathbf{Y}_j = \emptyset$ for any $i \ne j$.
\end{proportion}
\begin{proof}
	Suppose there is one message $X_u$, which can be decoded given either $\mathbf{V}_i$ or $\mathbf{V}_j$. 
	Additionally, we assume $i \not\in \mathbf{Y}_j$ and $j \not\in \mathbf{Y}_i$.
	To construct a new coding scheme, we can remove $X_u$ from the coded transmissions which can be used to decoded $X_u$ and $X_j$ given $X_{\mathbf{V}_j}$.
	After the modification, $X_u$ can still be decoded given $X_{\mathbf{V}_i}$ and $X_{\mathbf{Y}_j \setminus \{u\}}$ can still be decoded given $X_{\mathbf{V}_j}$.
	And the total number of required transmissions does not increase.
\end{proof}

\begin{lemma}
	If $K\le N^2+N+M$, the minimum number of required transmissions is $K-M$. 
\end{lemma}
\begin{proof}
	According to Proportion~\ref{pp:NumOfMes}, for each $\mathbf{V}_i\ne \emptyset$ ($i\in \mathbf{W}_0$), we have $|\mathbf{Y}_i| \ge N$.
	According to Proportion~\ref{pp:Y}, we can assume that for any $i\ne j \in \mathbf{W}_0$, $\mathbf{Y}_i \cap \mathbf{Y}_j = \emptyset$, then we have
	\begin{align}
	\sum_{i \in \mathbf{W}_0} |\mathbf{V}_i| + |\mathbf{Y}_i| = M+N^2
	\end{align}
	Thus, if $K < N^2 + M$, there must exist $l \in \mathbf{W}_0$ such that $\mathbf{Y}_l = \emptyset$.
	In such cases, $i^* = l$ and $|\cup_{j\in \mathbf{W}_0\setminus i^*} \mathbf{V}_j| = M$.
	That means all conditionally entropies in the summation are zero, except the first term $H(X_{\mathbf{S}_0}) = ML$, and the total number of required transmissions is $K-M$.
	
	If $N^2+M \le  K \le N^2+M+N$, there are enough messages such that for any $i,j\in \mathbf{W}_0$, $\mathbf{Y}_i \cap \mathbf{Y}_j = \emptyset$ and $\mathbf{V}_i \ne \emptyset$.
	However, since $|\mathbf{W}_0\cup\mathbf{S}_0\cup\mathbf{U}_0| \ge M+N$, the number of messages for $X_{\mathbf{W}_1}$, $X_{\mathbf{S}_1}$ and $X_{\mathbf{U}_1}$ is at most $N$.
	The side information that can be used to decode the new message which was added in $\mathbf{W}_1$ is zero.
	Hence, we have 
	\begin{align}
	H(X_{\mathbf{S}_1}|A^{[\mathbf{W}_1,\mathbf{S}_1]},Q^{[\mathbf{W}_1,\mathbf{S}_1]},X_{\mathbf{W}_0\cup\mathbf{S}_0\cup\mathbf{U}_0}) = 0.
	\end{align}
	Otherwise, if $H(X_{\mathbf{S}_1}|X_{\mathbf{W}_0\cup\mathbf{S}_0\cup\mathbf{U}_0}) > 0$, we have $|\mathbf{W}_1\cup\mathbf{U}_1| <N$.
	This means messages indexed by subset of $\mathbf{W}_1\cup\mathbf{S}_1\cup\mathbf{U}_1$ with size $N$ cannot be the indices of demand messages, which violates the privacy condition.
	Therefore, for $K\le N^2+N+M$,  the minimum number of required transmissions is $K-M$.  
\end{proof}

If $K \ge N^2+N+M$, it is possible to select $\mathbf{V}_i$'s ($i\in \mathbf{W}_0$) such that each conditional entropy in the summation can achieve their maximum. The number of required messages for $\mathbf{W}_0$, $\mathbf{S}_0$ and $\mathbf{U}_0$ is
\begin{align}
|\mathbf{W}_0\cup \mathbf{S}_0\cup \mathbf{U}_0| &=  |\mathbf{W}_0|+ |\mathbf{S}_0|+| \mathbf{U}_0|\\
&= N + M + N(N-1)\\
& = N^2+M
\end{align}

As we have shown, only when $i_t \not\in (\mathbf{W}_0^{t-1}\cup\mathbf{S}_0^{t-1}\cup\mathbf{U}_0^{t-1})$, the corresponding conditional entropy is positive.
And for each $i_t \not\in (\mathbf{W}_0^{t-1}\cup\mathbf{S}_0^{t-1}\cup\mathbf{U}_0^{t-1})$, there must have $N-1$ messages that can also be decoded given the new side information messages which are used for decoding $X_{i_t}$.

Hence, we have 
\begin{align}
T^* = \left\lceil\frac{K-M-N^2}{N+\lfloor \frac{M}{N} \rfloor}\right\rceil
\end{align}
And if $K-M-N^2-(T^8-1)( N+\lfloor \frac{M}{N} \rfloor) \le N$, there are at most $N$ messages left after using $X_{\mathbf{W}_0^{T^*-1}\cup\mathbf{S}_0^{T^*-1}\cup\mathbf{U}_0^{T^*-1}}$. Hence, they should be sent separately.
Therefore, we have
\begin{align}
R = \lim_{L\to \infty}\frac{D}{L} \ge K-M - (T^*-1)^+\lfloor \frac{M}{N} \rfloor + (K-M-N^2-(T^*-1)^+( N+\lfloor \frac{M}{N} \rfloor)- N)^+
\end{align}
which can be shown to be equivalent to~\eqref{eq:main}. Therefore we prove the converse of the minimum number of required transmissions in an alternative way.

\section*{ACKNOWLEDGMENT}
This work was supported in part by the Swiss National Science Foundation under Grant 169294.

\bibliographystyle{IEEEtran}
\bibliography{edic}

\end{document}